\documentclass{article}
\usepackage{graphicx}
\usepackage[english]{babel}
\usepackage{amsthm}
\usepackage{amsmath, amssymb, pxfonts}
\usepackage{cases}
\usepackage{booktabs}

\newtheorem{theorem}{Theorem}[section]
\newtheorem{lemma}[theorem]{Lemma}
\newtheorem{prop}[theorem]{Proposition}
\newtheorem{cor}[theorem]{Corollary}

\theoremstyle{definition}
\newtheorem{definition}[theorem]{Definition}
\newtheorem{remark}[theorem]{Remark}

\title{Robust Volatility Index Calculation with OTM Option-implied Probability}
\author{Masaaki Fukasawa and Shunta Murayama\\
{\small Graduate School of Engineering Science, The University of Osaka}}
\date{}

\begin{document}
\maketitle

\section{Introduction}

In financial markets, accurately measuring the risk of future fluctuations in asset prices is of paramount importance. Studies such as Carr and Madan \cite{Carr} have shown that the expected value of the quadratic variation of log prices can be expressed as an integral of European option prices over a continuum of strikes. This has led to the widespread estimation of model-free volatility (implied variance). However, this theoretical calculation assumes that options are continuously traded across all strike prices, which creates a fundamental gap with real-world market environments where options are only traded at discrete strikes. How to appropriately address this gap and robustly estimate volatility is a crucial issue for both practitioners and academics, and is the primary objective of this paper.

As a practical workaround, many volatility indices—such as the VIX, calculated by the CBOE for the S\&P 500 index—approximate the continuous integral as a Riemann sum using discrete strike prices observed in the market. However, this discrete approximation not only leads to an underestimation of volatility due to truncation errors at the tails, but also harbors a critical risk: an arbitrage-free market model (or a valid risk-neutral probability measure) corresponding to the observed discrete price model may not exist.

To overcome these discretization and truncation problems, methods that interpolate implied variance using cubic polynomials (splines) have been employed in indices such as the VXJ, a volatility index calculated by The University of Osaka for Nikkei 225 options \cite{VXJ}. Nevertheless, it has been pointed out that even with this interpolation method, in highly illiquid market environments or under extreme market stress, issues arise where no-arbitrage conditions are violated—such as the interpolated variance or the implied probability density becoming negative.

Focusing on the fact that volatility indices are primarily calculated from the prices of out-of-the-money (OTM) options, this paper proposes a novel method for constructing a continuous European option pricing function that is consistent with the bid-ask spreads of observed OTM options and strictly satisfies arbitrage-free conditions (such as monotonicity and convexity). Although previous studies have attempted to construct arbitrage-free option pricing functions from bid-ask spreads, the construction method proposed in this paper requires fewer market parameters than existing methods. This makes it possible to robustly calculate volatility indices while maintaining theoretical consistency, even in markets with extremely low liquidity.

The remainder of this paper is organized as follows. Section 2 reviews the background of this study: Section 2.1 summarizes the theoretical properties that option prices must satisfy in an arbitrage-free market, and Section 2.2 outlines the theory of model-free implied volatility along with conventional calculation methods. Section 3 presents the main results regarding the construction of the option pricing function. Specifically, Section 3.1 introduces the proposed price construction algorithm and provides its theoretical guarantees, while Section 3.2 discusses the advantages of the proposed method by comparing it with conventional algorithms. Section 4 conducts empirical analyses using actual market data. Section 4.1 compares the proposed method with the VIX using S\&P 500 (SPX) option data, and Section 4.2 presents a case study demonstrating the robustness of the proposed method. Finally, Section 4.3 discusses practical treatments for situations where arbitrage opportunities exist in the observed market data.

\section{Preliminaries}
\subsection{Arbitrage-Free Market}

This section specifies the market model and the robust, arbitrage-free framework for our analysis. We consider a model-free financial market over a fixed time horizon $[0, T]$, where $0$ represents the current time and $T$ is the maturity. We let $D$ be the price of a zero-coupon bond that matures at time $T$. Furthermore, let $F_0$ denote the current forward price of the underlying asset with maturity $T$. In this study, we focus exclusively on static portfolios, and we do not assume any specific probability distribution for the underlying asset price. Following the pointwise approach of \cite{Riedel}, we define arbitrage. 

Let $\Omega = [0, \infty)$ be the state space of the underlying asset price at maturity $T$, where each state $\omega \in \Omega$ corresponds to the terminal asset price, i.e., $S_T(\omega) = \omega$. The market allows for trading in $n$ derivative securities. To incorporate market illiquidity, we assume that each security $i$ ($i = 1, \dots, n$) is traded with a bid-ask spread, where $A_i$ and $B_i$ denote the ask and bid prices, respectively ($A_i \ge B_i$). A static portfolio is represented by a vector of positions $\pi = (\pi_1, \dots, \pi_n) \in \mathbb{R}^n$. The initial cost of establishing the portfolio $\pi$ is given by
\begin{align*}
C(\pi) = \sum_{i=1}^n \left( (\pi_i)^+ A_i - (\pi_i)^- B_i \right),
\end{align*}
where $(\pi_i)^+ = \max\{\pi_i, 0\}$ and $(\pi_i)^- = \max\{-\pi_i, 0\}$. At maturity $T$, the portfolio yields a terminal payoff $V(\omega) = \sum_{i=1}^n \pi_i f_i(\omega)$, where $f_i(\omega)$ is the payoff of the $i$-th security in state $\omega$. Taking the time value of money into account, the net terminal value of the portfolio at maturity is expressed as $V(\omega) - \frac{C(\pi)}{D}$. Under this robust framework, we define arbitrage as follows.

\begin{definition}\label{arbi}
    A static portfolio $\pi \in \mathbb{R}^n$ is called an \textit{arbitrage strategy} if its net terminal value satisfies the following two conditions:
    \begin{enumerate}
        \item $V(\omega) - \frac{C(\pi)}{D} \ge 0$ for all $\omega \in \Omega$.
        \item $V(\omega) - \frac{C(\pi)}{D} > 0$ for at least one $\omega \in \Omega$.
    \end{enumerate}
\end{definition}
\begin{definition}\label{noarbi}
    The market is defined as \textit{arbitrage-free} if it does not admit any arbitrage strategies.
\end{definition}
\begin{remark}
    Our framework is general enough to encompass frictionless markets where no bid-ask spreads exist. In such cases, we simply set the ask price equal to the bid price ($A_i = B_i = P_i$) for all $i$. The initial cost of the portfolio then naturally reduces to the standard linear pricing rule, $C(\pi) = \sum_{i=1}^n \pi_i P_i$.
\end{remark}

If an arbitrage strategy existed in the market, all investors would be able to profit with no initial cost and no risk, leading to the collapse of the market itself. Therefore, the assumption that the market is arbitrage-free is a crucial premise in our market model for calculating the volatility index. Next, we will examine the properties that European options must satisfy. In particular, assuming the interest rate is deterministic, the following proposition holds in an arbitrage-free market. This is well-known, and a proof can be found in standard textbooks such as \cite{hull2008}.
\begin{prop}\label{option_prop}
    Let us consider a market with no price spreads, where call and put options can be traded for any strike price $K\in[0,\infty)$. If we denote the trading prices of put and call options at each strike price $K$ as $P(K)$ and $C(K)$ respectively, then the following holds true when the market is arbitrage-free.
    \begin{enumerate}
        \item $P$ is a monotonically increasing, convex function that satisfies $P(0)=0$.
        \item $C$ is a monotonically decreasing, convex function.
        \item For any strike prices $K_1$ and $K_2$ with $K_1 < K_2$, the average rate of change of $P$ is bounded above by $D$:
        \[ \frac{P(K_2)-P(K_1)}{K_2-K_1} \le D \]
        \item For any strike prices $K_1$ and $K_2$ with $K_1 < K_2$, the average rate of change of $C$ is bounded below by $-D$:
        \[ \frac{C(K_2)-C(K_1)}{K_2-K_1} \ge -D \]
        \item $P(K)-C(K)=D(K-F_0)\quad \text{(\textit{put-call parity})}$
    \end{enumerate}
\end{prop}
The following corollary is an immediate consequence of this proposition.
\begin{cor}\label{corF0}
    The forward price, $F_0$, corresponds to the strike price at which the relative pricing of a call and a put option inverts.
\end{cor}

\subsection{Model-Free Implied Volatility}
In this section, we will explain the calculation method for the VIX, the volatility index calculated by the CBOE for the SPX, and examine the theoretical problems in its methodology. We let $F$ be the forward price process of the underlying asset, and we assume the existence of a probability measure $\mathbb{Q}$ under which $F$ is a martingale. 
Moreover, the interest rate is assumed to be deterministic. 
Using It\^o's formula for $f(x)=\log{x}$, we can compute the quadratic variation of the logarithm as follows.
\[
\langle{\log{F}}\rangle_T
=-2\log{(F_T/F_0)}+2\int_0^T\frac{\text{d}F_t}{F_t}
\]
We then take the expectation with respect to the measure $\mathbb{Q}$ and apply integration by parts to arrive at the following result.
\begin{align*}
    \mathbb{E}^\mathbb{Q}[\langle{\log{F}}\rangle_T]
    &=-2\mathbb{E}^\mathbb{Q}[\log{(F_T/F_0)}] \\
    &=2\int_0^{F_0}\mathbb{E}^\mathbb{Q}[(K-F_T)^+]\frac{\text{d}K}{K^2}+2\int_{F_0}^\infty\mathbb{E}^\mathbb{Q}[(F_T-K)^+]\frac{\text{d}K}{K^2} 
\end{align*}
The expectation terms appearing in the integral of this final term correspond to the expected payoffs of put and call options with respect to the strike price $K$, respectively. Therefore, this relationship allows us to calculate the expected quadratic variation from option prices observed in the market. However, two problems arise when performing this calculation in practice. First, only a finite number of strike prices are traded in the market. Second, actual transaction data include a bid-ask spread. These issues prevent us from executing the integration over the strike price $K$ at a single, spread-free price.
To overcome these issues, the VIX is calculated using the following approach. 
First, a strike price $K_\star$ is selected from the set of traded strikes such that it is sufficiently close to the forward price $F_0$. The integration interval is then modified, and the integral is approximated in the following manner.
\begin{align*}
    &2\int_0^{F_0}\mathbb{E}^\mathbb{Q}[(K-F_T)^+]\frac{\text{d}K}{K^2}+2\int_{F_0}^\infty\mathbb{E}^\mathbb{Q}[(F_T-K)^+]\frac{\text{d}K}{K^2} \\
    &=2\int_0^{K_\star}\mathbb{E}^\mathbb{Q}[(K-F_T)^+]\frac{\text{d}K}{K^2}+2\int_{K_\star}^\infty\mathbb{E}^\mathbb{Q}[(F_T-K)^+]\frac{\text{d}K}{K^2}+2\int_{K_\star}^{F_0}(K-F_0)\frac{\text{d}K}{K^2} \\
    &\approx 2\int_0^{K_\star}\mathbb{E}^\mathbb{Q}[(K-F_T)^+]\frac{\text{d}K}{K^2}+2\int_{K_\star}^\infty\mathbb{E}^\mathbb{Q}[(F_T-K)^+]\frac{\text{d}K}{K^2} - \left(\frac{F_0}{K_\star}-1\right)^2
\end{align*}
The remaining integral is then approximated via a Riemann sum, which is constructed from observed option prices. The price for each option used in the sum is taken as the midpoint of its bid-ask spread. This procedure yields a model-free estimate for the expected quadratic variation for each maturity $T$, denoted by $V(T)$, which is given by the following equation:
\[
V(T)=\frac{2}{D}\sum_{K}\frac{Q(K)\Delta{K}}{K^2}-\left(\frac{F_0}{K_\star}-1\right)^2
\]
In this context, $Q(K)$ denotes the price of the relevant out-of-the-money option (either a put or a call) for each observable strike price $K$. The VIX is then calculated by obtaining this $V(T)$ for two maturities, $T_1$ and $T_2$, and linearly interpolating these values to a target maturity of $T=30$ days.

We now turn to the issues arising from this methodology. The first issue is the truncation of the tails, which results from approximating an integral over an infinite interval with a Riemann sum. Since the integrand is always positive, this truncation leads to an underestimation of volatility. Such an underestimation has also been pointed out in studies such as \cite{Jiang2007}. Another issue is that by taking the midpoint of the bid-ask spread, the resulting option prices may no longer satisfy theoretical properties such as convexity. This is particularly observed when market liquidity is low and the bid-ask spread is wide. This creates a problem where the calculation is performed in a market model that effectively contains arbitrage opportunities, which in turn reduces the reliability of the estimate. In the following sections, we will describe a calculation method designed to overcome these problems.

\section{Construction of Option Prices}
\subsection{Arbitrage-free Interpolation and Extrapolation}
A difficulty in volatility estimation stems in part from the discrete nature of observable strike prices. In a hypothetical scenario where put and call prices are continuously available over $[0, \infty)$, the necessity of transforming the integration interval or approximating it via a Riemann sum is eliminated. Moreover, applying Corollary \ref{corF0} simplifies the calculation as shown below.
\begin{align*}
&2\int_0^{F_0}\mathbb{E}^\mathbb{Q}[(K-F_T)^+]\frac{\text{d}K}{K^2}+2\int_{F_0}^\infty\mathbb{E}^\mathbb{Q}[(F_T-K)^+]\frac{\text{d}K}{K^2}\\
&=\frac{2}{D}\int_0^\infty\min\{P(K),C(K)\}\frac{\text{d}K}{K^2}
\end{align*}
Notably, the consolidation of the integrals eliminates the need for $F_0$ in the calculation. While the standard VIX algorithm relies on determining $F_0$ via put-call parity, this process becomes challenging when the set of available strike prices is sparse, often due to limited market liquidity. This issue will be addressed further in the subsequent section. Accordingly, we now construct a continuous option price function from observed market data. As noted in the preceding section, a model based on bid-ask mid-prices can fail to be convex. Thus, our construction ensures that the price curve passes through the bid-ask spread at every strike price. This paper proposes an option price construction method derived from arbitrage-free conditions, founded on the principle that the supremum of linear functions is convex. 

We now focus on the construction of put option prices. Suppose that at the current time, put options on an underlying asset with maturity $T$ are available at a set of strike prices $0<K_1 < K_2 < \cdots < K_N$, and denote their respective ask and bid prices as $A_n$ and $B_n$ for each $K_n$. Let $D$ denote the price of a zero-coupon bond maturing at $T$. For convenience, we also set $K_0 = A_0 = B_0 = 0$. We now define the arbitrage-free property for the price system $(K_n, A_n, B_n)_{1 \le n \le N}$ as follows.
\begin{definition}
    The system $(K_n, A_n, B_n)_{1 \le n \le N}$ is said to be \textit{arbitrage-free} if the market, consisting of zero-coupon bonds and put options traded at strike prices $K_n$ with corresponding ask prices $A_n$ and bid prices $B_n$, is itself arbitrage-free.
\end{definition}
As the linear functions for taking the supremum, we will primarily use the lines connecting the ask prices. Therefore, for each $0 \le i < j \le n$, let $f_{i,j}$ denote the function of the line passing through the two points $(K_i, A_i)$ and $(K_j, A_j)$. That is, $f_{i,j}$ is a function of the following form:
\[
f_{i,j}(K)=\frac{A_j-A_i}{K_j-K_i}(K-K_i)+A_i\quad\text{for }K\in\mathbb{R}  
\]
Next, we define the set $\mathcal{L}$, which underlies the class of linear functions for taking the supremum, as follows.
\[
\mathcal{L}_0=\left\{ f_{i,j}\,;\,f_{i,j}(0)<0\text{ and } f_{i,j}(K_n)\le A_n \text{ for all } n=1,\dots,N\right\}
\]
The first condition, $f_{i,j}(0) < 0$, is derived from Proposition \ref{option_prop}. The second condition, $f_{i,j}(K_n) \le A_n$, stems from the objective of constructing an option price that does not exceed the ask prices. Furthermore, as stated in Proposition \ref{option_prop}, the arbitrage-free condition requires the slope of the option price to be less than or equal to D. To this end, we impose an additional restriction on the set of functions and define new functions as follows.
\begin{align*}
    \mathcal{L}_D &= \left\{ f_{i,j}\,;\, f^\prime_{i,j}\le D \right\}, \\
    \mathcal{L} &= \mathcal{L}_0 \cap \mathcal{L}_D, \\
    f^D_i(K) &=D(K-K_i)+A_i\quad\text{for }i\in\{1,\cdots,N\}, \\
    f^D(K) &= \min_{i\in\{1,\cdots,N\}} f^D_i(K).
\end{align*}
From the definition of these function sets, the following lemmas hold.
\begin{lemma}\label{lem0}
    Let $i, j, k, l$ be indices such that $f_{i,j}, f_{k,l} \in \mathcal{L}$ and $i < k$. Then, $f_{i,j}^\prime \le f_{k,l}^\prime$.
\end{lemma}
\begin{proof}
Since the derivatives $f_{i,j}^\prime, f_{k,l}^\prime$ are constant and 
$f_{i,j}, f_{k,l} \in \mathcal{L}$, we have
\begin{equation*}
    \begin{split}
       & A_i \geq f_{k,l}(K_i) = A_k + f^\prime_{k,l}(K_i-K_k),\\
       & A_k \geq f_{i,j}(K_k) = A_i  + f^\prime_{i,j}(K_k-K_i),  
    \end{split}
\end{equation*}
which implies 
\begin{equation*}
   f^\prime_{k,l}(K_k-K_i) \geq  A_k-A_i \geq f^\prime_{i,j}(K_k-K_i).
\end{equation*}
Therefore, if $i < k$, then $f_{i,j}^\prime \le f_{k,l}^\prime$.
\end{proof}
\begin{lemma}\label{lem1}
If there exist integers $i,j$ with $0<i<j<N$ such that $f_{i,j}\in\mathcal{L}_0$, then there exists an integer $k$ with
$j<k\le N$ such that $f_{j,k}\in\mathcal{L}_0$.
\end{lemma}
\begin{proof}
Suppose that $f_{i,j}\in\mathcal{L}_0$ with $0<i<j<N$. From the definition of $\mathcal{L}_0$, we have $f_{i,j}(K_k)\le A_k$ for any $k\in\{j+1,\cdots,N\}$, that is:
\[
\frac{A_j-A_i}{K_j-K_i}(K_k-K_j)+A_j\le A_k
\]
Thus, $f^\prime_{i,j}\le f^\prime_{j,k}$, which implies that $f_{j,k}(K_n)\le f_{i,j}(K_n)\le A_n$ for $n\le j$.
We define $k$ as follows:
\[k=\max\left\{\underset{k\in\{j+1,\cdots,N\}} {\operatorname{argmin}} \frac{A_k-A_j}{K_k-K_j}\right\}\]
By this definition, for any $n$ such that $j<n$,
\[\frac{A_k-A_j}{K_k-K_j}\le\frac{A_n-A_j}{K_n-K_j}\]
and thus $f_{j,k}(K_n)\le A_n$ holds. Therefore, for $k$ defined in this way, $f_{j,k}\in\mathcal{L}_0$.
\end{proof}
Here, we define the indices $I^\mathcal{L}, J^\mathcal{L}$ as follows.
\[ 
    I^\mathcal{L} = \min\{i \,;\, \exists j \text{ s.t. } i < j \le N, f_{i,j} \in \mathcal{L} \}, \quad
    J^\mathcal{L} = \max\{j \,;\, \exists i \text{ s.t. } 0 < i < j, f_{i,j} \in \mathcal{L}\}.
\]
From their definitions, $I^\mathcal{L}$ and $J^\mathcal{L}$ represent the minimum and maximum indices, respectively, of the strike prices used in the construction of $\mathcal{L}$. Furthermore, by applying Lemma \ref{lem1}, the following holds for $J^\mathcal{L}$.
\begin{lemma}\label{lemP1}
    $f^D=f^D_{J^\mathcal{L}}$.
\end{lemma}
\begin{proof}
It suffices to show that $f^D_{J^\mathcal{L}}(K_n) \le A_n$ holds for all $n$. First, we consider the case where $J^\mathcal{L} < N$. From Lemma \ref{lem1} and the definition of $J^\mathcal{L}$, there exists some $j > J^\mathcal{L}$ such that $f_{J^\mathcal{L},j} \in \mathcal{L}_0$ and $f^\prime_{J^\mathcal{L},j} > D$. Therefore, for any $n > J^\mathcal{L}$, we have $f^D_{J^\mathcal{L}}(K_n) \le f_{J^\mathcal{L},j}(K_n) \le A_n$. Furthermore, from the definition of $J^\mathcal{L}$, there exists some $i < J^\mathcal{L}$ such that $f_{i,J^\mathcal{L}} \in \mathcal{L}_0$ and $f^\prime_{i,J^\mathcal{L}} \le D$. Consequently, for $n < J^\mathcal{L}$, we also obtain $f^D_{J^\mathcal{L}}(K_n) \le f_{i,J^\mathcal{L}}(K_n) \le A_n$. Since $f^D_{J^\mathcal{L}}(K_{J^\mathcal{L}}) = A_{J^\mathcal{L}}$, this completes the proof for the case where $J^\mathcal{L} < N$. The case where $J^\mathcal{L} = N$ also follows from the above argument.
\end{proof}
\begin{lemma}\label{lemP2}
    If there exists $f_{i,j}\in\mathcal{L}$ with $0<i<j<J^\mathcal{L}$, then there exists some $k$ with $j<k\le J^\mathcal{L}$ such that $f_{j,k}\in\mathcal{L}$
\end{lemma}
\begin{proof}
    By Lemma \ref{lem1}, there exists some $k$ such that $f_{j,k} \in \mathcal{L}_0$. Assuming $f_{j,k} \notin \mathcal{L}$, we have $f^\prime_{j,k} > D$ and $f_{j,k}(K_n) \le A_n$ for all $n$. Therefore, noting that $f^D_j(K_j) = f_{j,k}(K_j)$, it follows that $f^D_j(K_{J^\mathcal{L}}) < f_{j,k}(K_{J^\mathcal{L}}) \le A_{J^\mathcal{L}} = f^D_{J^\mathcal{L}}(K_{J^\mathcal{L}})$. From Lemma \ref{lemP1}, we have $f^D = f^D_{J^\mathcal{L}}$, but this contradicts the minimality of $f^D$. Thus, $f_{j,k} \in \mathcal{L}$. The condition $k \le J^\mathcal{L}$ is obvious from the definition of $J^\mathcal{L}$.
\end{proof}
Next, the following lemmas follow from the arbitrage-free condition of the bid and ask prices.
\begin{lemma}\label{lem2}
If the system $(K_n, A_n, B_n)_{1 \le n \le N}$ is arbitrage-free, then for any integers $i,j,k$ with $0\le i<j<k\le N$, it holds that $f_{i,k}(K_j)>B_j$.
\end{lemma}
\begin{proof}
We show that an arbitrage strategy exists when $f_{i,k}(K_j)>{B_j}$ does not hold.
From the ordering of the strike prices, we can choose $\lambda\in(0,1)$ so that $\lambda{K_i}+(1-\lambda)K_k=K_j$.
Furthermore, assuming $f_{i,k}(K_j)\le B_j$, the following inequality holds:
\[\frac{A_k-A_i}{K_k-K_i}(K_j-K_k)+A_k\le B_j\]
From this and the definition of $\lambda$, we obtain $\lambda A_i+(1-\lambda)A_k\le B_j$. Now, let us consider an investment strategy: buy $\lambda$ units of the put option with strike $K_i$ and $(1-\lambda)$ units of the put option with strike $K_k$, and short sell one unit of the put option with strike $K_j$. Let $V$ be the terminal payoff of this strategy. By assumption, the profit at time 0 is $B_j-\{\lambda A_i+(1-\lambda)A_k\}\ge0$. Furthermore, the profit at time $T$ can be evaluated by the convexity of $(x)^+$ as follows:
\begin{align*}
&\lambda(K_i-S_T)^++(1-\lambda)(K_k-S_T)^+-(K_j-S_T)^+\\
&\ge ((\lambda{K_i}+(1-\lambda)K_k)-S_T)^+-(K_j-S_T)^+=0
\end{align*}
Moreover, this profit takes a strictly positive value when $S_T \in (K_i, K_k)$. Thus, this investment strategy constitutes an arbitrage strategy, and the lemma is proven by contradiction.
\end{proof}
\begin{lemma}\label{lemD}
If the system $(K_n, A_n, B_n)_{1 \le n \le N}$ is arbitrage-free, then we have $f^D_i(K_j)\ge B_j$ for any $i<j$.
\end{lemma}
\begin{proof}
    We assume that $f^D_i(K_j) < B_j$ for some $i < j$ to derive a contradiction. First, from the definition of $f^D_i$, we have $D(K_j - K_i) + A_i < B_j$. Therefore, we can construct an investment strategy that yields a non-negative cash flow at time $0$ as follows: buy one unit of the put option with strike $K_i$ and $K_j - K_i$ units of zero-coupon bonds, and short sell one unit of the put option with strike $K_j$. The payoff of this strategy at time $T$ can be evaluated as follows:
    \begin{align*} 
    &K_j-K_i+(K_i-S_T)^+ -(K_j-S_T)^+ \\
    &=(K_i-S_T)^+ -(K_i-S_T) -((K_j-S_T)^+ -(K_j-S_T)) \\ 
    &=(K_i-S_T)^- - (K_j-S_T)^- \ge 0 
    \end{align*} 
    Since this constitutes an arbitrage strategy, it contradicts the no-arbitrage condition of the system.
\end{proof}
\begin{lemma}\label{lem3}
    For each $0 <  i<j\le N$, let 
    \[
    g_{i,j}(K) =  \frac{A_j-B_i}{K_j-K_i}(K-K_j)+A_j\quad\text{for }K\in\mathbb{R}.
    \]
    If the system $(K_n, A_n, B_n)_{1 \le n \le N}$ is arbitrage-free, then $g_{i,j}(0)<0$ holds.
\end{lemma}
\begin{proof}
    From the arbitrage-free assumption, we apply Lemma \ref{lem2} with $i=0, j=i,$ and $k=j$ to obtain $f_{0,j}(K_i)>B_i$.
    Hence, $A_jK_i>B_iK_j$ and so,
    \begin{align*} 
        g_{i,j}(0)
        &=-\frac{A_j-B_i}{K_j-K_i}K_j+A_j \\
        &=\frac{1}{K_j-K_i}(-(A_j-B_i)K_j+A_j(K_j-K_i)) \\
        &=\frac{1}{K_j-K_i}(-A_jK_i+B_iK_j)<0
    \end{align*}
    as claimed.
\end{proof}

By applying these lemmas, the put option price can be constructed as follows.
\begin{theorem}\label{thm1}
Suppose that $\mathcal{L}\neq\emptyset$ and that the system $(K_n, A_n, B_n)_{1 \le n \le N}$ is arbitrage-free.
Define $f_0$, $\tilde{\mathcal{L}}$ and $p$ by
\begin{align*} 
f_0(K) &= \left(\min_{1\le i<I^\mathcal{L}}\frac{A_{I^\mathcal{L}}-B_i}{K_{I^\mathcal{L}}-K_i}\right)(K-K_{I^\mathcal{L}})+A_{I^\mathcal{L}} \quad (I^\mathcal{L}\ge2),\\
\tilde{\mathcal{L}} &= 
\begin{cases}
\mathcal{L}\cup\{f^D,0\} & \text{if }I^\mathcal{L}=1\,\text{or }f_0^{\prime}>\min_{f\in\mathcal{L}}f^{\prime},\\
\mathcal{L}\cup\{f_0,f^D,0\} & \text{otherwise},
\end{cases}
\\
p & =\max_{f\in\tilde{\mathcal{L}}}f
\end{align*}
Then, the function $p$ is a non-decreasing convex function 
satisfying the following properties:
\begin{enumerate}
    \item[a)] There exists $\epsilon > 0$ such that
     $p(K) = 0$ for all $K \in [0,\epsilon]$.
    \item[b)] For any strike prices $K_1$ and $K_2$ with $K_1 < K_2$, the average rate of change of $p$ is bounded above by $D$:
        \[ \frac{p(K_2)-p(K_1)}{K_2-K_1} \le D \]
    \item[c)] $B_n \leq p(K_n) \leq A_n$ for all $n = 1,\dots,N$.
\end{enumerate}
\end{theorem}

\begin{proof}
Since $f(0) < 0$ for any $f \in \mathcal{L}$
and $f_0(0) < 0$ by Lemma~\ref{lem3}, we have 
 $f(0) < 0$ for any $f \in \tilde{\mathcal{L}}\setminus \{0\}$.
 This implies a), and also that $f^\prime \geq 0$ for any $f \in \tilde{\mathcal{L}}$.
 Therefore, 
every element of $\tilde{\mathcal{L}}$ is a non-decreasing linear function, which implies that $p$ is a non-decreasing convex function.

To show b), it suffices to prove that the slope of every function in $\tilde{\mathcal{L}}$ is less than or equal to $D$. For functions other than $f_0$, this is evident from the construction of $\mathcal{L}$. Furthermore, the slope of $f_0$ is also bounded by $D$ due to the condition required for its inclusion in $\tilde{\mathcal{L}}$. Thus, it is clear that b) is satisfied.

Now, we show c)  by proving $B_n\le{p(K_n)}$ and $p(K_n)\le{A_n}$ individually.
In order to establish $B_n\le{p(K_n)}$, we show that for each $n\in\{1,\cdots,N\}$, 
there exists $f\in\tilde{\mathcal{L}}$ such that $B_n\le{f(K_n)}$. 
Accordingly, we distinguish the following cases for $n$.
\begin{enumerate}
    \item[(i)] $n<I^\mathcal{L}$.
    \item[(ii)] $n>J^\mathcal{L}$
    \item[(iii)] There exists $f \in \tilde{\mathcal{L}}$ such that 
    $f(K_n) = A_n$.
    \item[(iv)] There exist two integers $l,m$ with $l<n<m$ such that $f_{l,m}\in\tilde{\mathcal{L}}$.
\end{enumerate}
Assume that $n$ satisfies neither (i), (ii) nor (iii).
In this case, since $I^\mathcal{L} \le n$, there exists at least one $i$ such that $i \le n$ and $f_{i,j} \in \mathcal{L}$ for some $j$. Let $i$ be the maximum of such indices.
We then take an index $j$ such that $f_{i,j} \in \mathcal{L}$.
If we assume $j < n$, then by Lemma \ref{lemP2}, there exists  $k>j$ such that $f_{j,k} \in \mathcal{L}$.
This contradicts the maximality of $i$. Therefore, we must have $n < j$.
Thus, in this case, $n$ satisfies (iv), which shows that every $n$ must satisfy one of (i), (ii), (iii) or (iv).

In Case (i), $f_0$ is, by definition, a function such that $f_0(K_n) \ge B_n$ for all $n < I^\mathcal{L}$. Even if $f_0 \notin \tilde{\mathcal{L}}$,
 in light of Lemma \ref{lem0}, there exists some $j$ such that $f_{I^\mathcal{L}, j} \in \tilde{\mathcal{L}}$ and $f_0^\prime > f_{I^\mathcal{L}, j}^\prime$.
Noting that $f_0(K_{I^\mathcal{L}}) = f_{I^\mathcal{L}, j}(K_{I^\mathcal{L}})$, we have $f_{I^\mathcal{L}, j}(K_n) \ge B_n$. In either scenario, there is $ f\in \tilde{\mathcal{L}}$ such that $f(K_n)\geq B_n$.

Case (ii) follows from Lemma \ref{lemD}.

Case (iii) is straightforward, as the arbitrage-free condition implies $B_n \le A_n$.

Case (iv) is a direct consequence of Lemma \ref{lem2}.
We have thus demonstrated that the price function $p$ does not fall below any of the bid prices.

Finally, we show that $p(K_n) \le A_n$.
By definition, the elements of $\mathcal{L}$ do not exceed any ask price, and furthermore, from its definition, $f^D(K_n) \le f^D_n(K_n) = A_n$ holds for all $n$, so it is sufficient to show that $f_0$ does not exceed the ask prices when
$f_0 \in \tilde{\mathcal{L}}$.
First, from the condition for $f_0 \in \tilde{\mathcal{L}}$ and Lemma \ref{lem0}, there exists some $j$ such that $f_{I^\mathcal{L},j} \in \mathcal{L}$ and $f_0^\prime \le f_{I^\mathcal{L},j}^\prime$ holds.
Also, since $f_0(K_{I^\mathcal{L}}) = f_{I^\mathcal{L},j}(K_{I^\mathcal{L}})$, $f_0(K_n) \le A_n$ holds for $n$ such that $I^\mathcal{L} \le n$.
Next, assume that $f_0(K_n) > A_n$ holds for some $n < I^\mathcal{L}$, that is:
\[ f_0^\prime(K_n-K_{I^\mathcal{L}})+A_{I^\mathcal{L}}>A_n \]
Therefore, we obtain $f_0^\prime \le f_{n,I^\mathcal{L}}^\prime$.
We fix $n$ (among such $n$) that maximizes the slope of $f_{n,I^\mathcal{L}}$.
That is, we take $n$ such that $f_{n,I^\mathcal{L}}$ is given as follows:
\[ f_{n,I^\mathcal{L}}(K)=\left(\max_{1\le n<I^\mathcal{L}}\frac{A_{I^\mathcal{L}}-A_n}{K_{I^\mathcal{L}}-K_n}\right)(K-K_{I^\mathcal{L}})+A_{I^\mathcal{L}} \]
From $f_{n,I^\mathcal{L}}(K_{I^\mathcal{L}})=f_0(K_{I^\mathcal{L}})$ and $f_0^\prime \le f_{n,I^\mathcal{L}}^\prime$, it follows that $f_{n,I^\mathcal{L}}(K_m)<f_0(K_m)$ for $m<I^\mathcal{L}$. Therefore, $f_{n,I^\mathcal{L}}(0)<f_0(0)<0$.
Furthermore, from the maximality of the slope, the following holds for $1\le m<I^\mathcal{L}$:
\[ \frac{A_{I^\mathcal{L}}-A_m}{K_{I^\mathcal{L}}-K_m}\le\frac{A_{I^\mathcal{L}}-A_n}{K_{I^\mathcal{L}}-K_n}. \]
Therefore, $f_{n,I^\mathcal{L}}(K_m) \le f_{m,I^\mathcal{L}}(K_m) = A_m$ holds.
Also, let $j$ be an index such that $f_{I^\mathcal{L},j}\in\mathcal{L}$. Then, from $f_{I^\mathcal{L},j}(K_n)\le A_n$, we have $f_{n,I^\mathcal{L}}^\prime \le f_{I^\mathcal{L},j}^\prime$. Therefore, since $f_{n,I^\mathcal{L}}(K_{I^\mathcal{L}})=f_{I^\mathcal{L},j}(K_{I^\mathcal{L}})$, $f_{n,I^\mathcal{L}}(K_m)\le A_m$ also holds for $m$ such that $I^\mathcal{L}\le m$.
As a result, $f_{n,I^\mathcal{L}}\in\mathcal{L}$, which contradicts the minimality of $I^\mathcal{L}$.
Therefore, $f_0(K_n)\le A_n$ holds for $n<I^\mathcal{L}$ as well.
Thus, $f_0(K_n) \le A_n$ for all $n$, and all assertions of the theorem have been proven.
\end{proof}

This theorem requires the condition $\mathcal{L} \neq \emptyset$. However, even in the case where $\mathcal{L} = \emptyset$, we can construct a function that satisfies the same conditions. To this end, we assume that $\mathcal{L} = \emptyset$. Furthermore, similarly to $f^D$, we define below a function with slope $D$ that bounds all bid prices from above.
\begin{align*}
    &g^D_i(K)=D(K-K_i)+B_i \quad \text{for }i\in\{1,\cdots,N\} \\
    &g^D(K)=\max_{i\in\{1,\cdots,N\}}g^D_i(K) 
\end{align*}
Here, if $g^D \le f^D$, we can see that adopting $g^D$ as the price ensures it lies between the bid and ask prices. Furthermore, its behavior near the origin can be justified by the following lemma.
\begin{lemma}
    If the system $(K_n, A_n, B_n)_{1 \le n \le N}$ is arbitrage-free, then $g^D_i(0) < 0$ for all $i$
\end{lemma}
\begin{proof}
    Suppose that $g^D_i(0) \ge 0$; we will show the existence of an arbitrage strategy. Since $g^D_i(0) = -DK_i + B_i$, we obtain $B_i \ge DK_i$. Therefore, if we consider the following strategy, the cash flow at time $0$ is non-negative: Purchase $K_i$ units of a zero-coupon bond and short sell $1$ unit of a put option with strike price $K_i$. Then, the payoff at time $T$ can be evaluated as follows:
    \[ K_i-(K_i-F_T)^+ \ge K_i-K_i = 0 \]
    Thus, this constitutes an arbitrage strategy, which contradicts the assumption of arbitrage-free markets. This completes the proof.
\end{proof}
Next, we prepare for the construction of the function in the case where $f^D < g^D$. First, we define the indexes $I,J$ such that $f^D_J = f^D,g^D_I=g^D$. Then, the next lemma holds.
\begin{lemma}\label{lemAd1}
    Assuming that $(K_n,A_n,B_n)_{1\le n\le N}$ is arbitrage-free and $f^D < g^D$, we have $I < J$.
\end{lemma}
\begin{proof}
    We prove this by contradiction. Suppose that $J < I$. Then we have $f^D(K_I) < g^D(K_I) = B_I$, which contradicts Lemma \ref{lemD}. Since it is clear that $I \neq J$, we conclude that $I < J$.
\end{proof}
For this $J$, we define the functions $f_1$ and $f_2$ as follows:
\begin{align*}
    f_1(K)=\left(\min_{1\le i<J}\frac{A_J-B_i}{K_J-K_i}\right)(K-K_J)+A_J \\
    f_2(K)=\left(\max_{J< j\le N}\frac{B_j-A_J}{K_j-K_J}\right)(K-K_J)+A_J
\end{align*}
The function $f_1$ is designed to exceed the bid prices at all strike prices $K_n$ for $n < J$, while $f_2$ is designed to exceed the bid prices at all strike prices $K_n$ for $J < n$. Therefore, by taking the maximum of these functions, we can construct a function that exceeds the bid prices. However, if $f^\prime_2 < f^\prime_1$, $f_1$ alone is sufficient; moreover, it is preferable to use only $f_1$ to avoid exceeding the ask prices. The construction of the function using these is justified by the following theorem.
\begin{theorem}
Assume that $(K_i, A_i, B_i)_{1 \le i \le N}$ is arbitrage-free and $\mathcal{L} = \emptyset$. We define the function $\tilde{p}(K)$ as follows:
\[
\tilde{p}(K) = 
\begin{cases}
\max\{g^D(K),0\} & \text{if } g^D \le f^D, \\
\max\{f_1(K), f_2(K),0\} & \text{if } f^D < g^D \text{ and } f^\prime_1 < f^\prime_2, \\
\max\{f_1(K),0\} & \text{otherwise}.
\end{cases}
\]
Then, $\tilde{p}$ is a non-decreasing convex function satisfying the following conditions:
\begin{itemize}
\item[a)] There exists $\epsilon > 0$ such that $\tilde{p}(K) = 0$ for all $K \in [0,\epsilon]$.
\item[b)] For any strike prices $K_1$ and $K_2$ with $K_1 < K_2$, the average rate of change of $\tilde{p}$ is bounded above by $D$:
\[ \frac{\tilde{p}(K_2)-\tilde{p}(K_1)}{K_2-K_1} \le D \]
\item[c)] $B_n \leq \tilde{p}(K_n) \leq A_n$ for all $n = 1,\dots,N$.
\end{itemize}
\end{theorem}
\begin{proof}
It is clear that the conditions are satisfied when $\tilde{p} = \max\{g^D(K), 0\}$. Therefore, we consider the other cases. When $f^\prime_1 < f^\prime_2$, noting that $f_1(K_J) = f_2(K_J)$, we have $f_2(0) < f_1(0)$. Therefore, for condition a), it is sufficient to show that $f_1(0) < 0$. Furthermore, since this function satisfies $f_1 = g_{i,J}$ for some $i$ by definition, we have $f_1(0) < 0$ from Lemma \ref{lem3}.

Next, for condition b), it suffices to show that $f_1^\prime \le D$ and $f_2^\prime \le D$. First, we prove this for $f_1$. From Lemma \ref{lemAd1} and $f^D < g^D$, we obtain $f^D(K_I) < B_I$, that is, $g^\prime_{I,J} < D$. Therefore, applying Lemma \ref{lemAd1} again and the definition of $f_1$, we have $f_1^\prime \le f^\prime_{I,J} < D$.
For $f_2$, this can be shown by using Lemma \ref{lemD} similarly to the proof of Lemma \ref{lemAd1}. Therefore, we can see that condition b) is satisfied.

Finally, for condition c), we will show that $\tilde{p}(K_n) \ge B_n$ and $\tilde{p}(K_n) \le A_n$ hold for all $n$. First, regarding $\tilde{p}(K_n) \ge B_n$, since $f_1$ was defined to be greater than or equal to $B_n$ for all $K_n$ with $n < J$, and $f_2$ was defined to be greater than or equal to $B_n$ for all $K_n$ with $J < n$, it follows that $\tilde{p}(K_n) \ge B_n$ holds for all $n$ in the case where $f^\prime_1 < f^\prime_2$. As for the case where $f^\prime_1 \ge f^\prime_2$, since $f_1(K_n) \ge f_2(K_n) \ge B_n$ holds for any $n$ with $J < n$, we conclude that $\tilde{p}(K_n) \ge B_n$ holds for all $n$ in this case as well.
Regarding $\tilde{p} \le A_n$, noting that $\tilde{p}(K) = f_2(K)$ implies $K_J \le K$, it is sufficient to show that $f_1(K_n) \le A_n$ for all $n$, and that $f_2(K_n) \le A_n$ for all $n$ such that $J < n$. First, for $n > J$, the minimality of $f^D$ implies $f^D(K_n) \le f^D_n(K_n) = A_n$. Noting that $\max\{f^\prime_1, f^\prime_2\} \le D$, we have $\max\{f_1(K_n), f_2(K_n)\} \le f^D(K_n) \le A_n$. To show that $f_1(K_n) \le A_n$ for any $n < J$, we assume the existence of an $n < J$ such that $f_1(K_n) > A_n$, and we will show that this contradicts $\mathcal{L} = \emptyset$. Among such $n$, we fix the one that maximizes the slope of $f_{n,J}$. By its definition, $f_{n,J}(K_m) \le A_m$ holds for any $m < J$. Furthermore, noting that $f^D(K_n) \le A_n$ (which implies $f^\prime_{n,J} \le D$) by the minimality of $f^D$, and that $f^D(K_m) \le A_m$ for any $m > J$, we obtain $f_{n,J}(K_m) \le f^D(K_m) \le A_m$. From the above, we see that $f_{n,J}^\prime \le D$ and $f_{n,J}(K_m) \le A_m$ hold for all $m$. Noting that $f_1^\prime < f_{n,J}^\prime$ from the assumption $f_1(K_n) > A_n$, we have $f_{n,J}(0) < f_1(0) < 0$. This ultimately means $f_{n,J} \in \mathcal{L}$, which is a contradiction. Therefore, by contradiction, $f_1(K_n) \le A_n$ for all $n < J$. Finally, we conclude the proof by noting that $f_1(K_J) = f_2(K_J) = A_J$.
\end{proof}

These theorems allow us to construct a put option price function that is consistent with the market and satisfies the arbitrage-free condition, with
\begin{equation*}
    \int_0^x \frac{p(K)}{K^2}\mathrm{d}K < \infty
\end{equation*}
for any $x>0$.
However, one problem still remains in the calculation using this option price.
That is, the volatility estimation can differ significantly between two market datasets that are both arbitrage-free and differ only slightly from each other. Since out-of-the-money (OTM) put options correspond to strike prices smaller than the forward price, the integration of $P(K)/K^2$ must be performed, particularly near $K=0$.
In this case, a slight difference in the price data near the lowest strike prices can lead to two different scenarios, both arbitrage-free: one where the y-intercept of $f_{1,2}$ is slightly negative, and another where it is zero or slightly positive. In the former case, this requires integrating a function of the order $O(K^{-1})$ (as $P(K) \propto K$) near $K=0$, which results in a very large value. In the latter case, $f_{1,2}$ no longer satisfies the condition for being an element of $\mathcal{L}$, causing the shape of the price function to change, and the integral does not become as large.
As a result, the volatility index cannot be computed robustly.

To avoid this problem, we devise a slightly modified method for constructing option prices. The aforementioned problem arises when $f_0$ is not included in $\tilde{\mathcal{L}}$. This is because $f_0$, by its definition, governs the function's values in the deepest OTM price range. As shown in the proof of the theorem, the y-intercept of $f_0$ must be negative in an arbitrage-free market; conversely, a non-negative (zero or positive) intercept implies an arbitrage opportunity in the market. Furthermore, in such a market, a market-consistent convex function satisfying $p(0)=0$ does not exist. Therefore, if $f_0$ is included, the volatility can be robustly estimated in the sense described above.

From this point onward, we construct a put option price that can be evaluated more robustly by modifying the class of functions to ensure $f_0$ is always included. In Theorem \ref{thm1}, $f_0$ was included only when $I^\mathcal{L} > 1$ and the slope of $f_0$ was larger compared to the slopes of the functions in $\mathcal{L}$. Therefore, we define the class of functions and its corresponding index as follows.
\begin{align*} 
&\mathcal{M}_0=\left\{ f_{i,j} \,;\, \exists m < i \text{ s.t. } f_{i,j}(K_m) < B_m \text{ and } \forall n, f_{i,j}(K_n) \le A_n\right\} \\
&\mathcal{M} = \mathcal{M}_0 \cap \mathcal{L}_D\\
&I^\mathcal{M} = \min\{i \,;\, \exists j \text{ s.t. } i < j \le N, f_{i,j} \in \mathcal{M} \}, \quad
J^\mathcal{M} = \max\{j \,;\, \exists i \text{ s.t. } 0 < i < j, f_{i,j} \in \mathcal{M}\}
\end{align*}
Suppose $f_{i,j} \in \mathcal{M}$. Then, from the conditions for $\mathcal{M}$, there exists some $m$ such that $1 \le m < i$ and the following holds.
\[ \frac{A_j-A_i}{K_j-K_i}(K_m-K_i)+A_i <B_m \]
Therefore, $g_{m,i}^\prime < f_{i,j}^\prime$. Given that $f_{i,j}(K_i) = g_{m,i}(K_i)$, it follows from Lemma \ref{lem3} that $f_{i,j}(0) < g_{m,i}(0) < 0$.
This shows that $f_{i,j} \in \mathcal{L}$, and consequently, $\mathcal{M} \subset \mathcal{L}$.
Furthermore, from the first condition, we have $I^\mathcal{M} > 1$.
The same condition and the definition of $I^\mathcal{M}$ also imply that there exist $j > I^\mathcal{M}$ and $k < I^\mathcal{M}$ such that $f_{I^\mathcal{M},j} \in \mathcal{M}$ and $f_{I^\mathcal{M},j}(K_k) < B_k$. Therefore, the following theorem is justified.
\begin{theorem}\label{thm2}
    We assume that $\mathcal{M} \neq \emptyset$ and that the system $(K_i, A_i, B_i)_{1 \le i \le N}$ is arbitrage-free.
   Define $f_0$, $\tilde{\mathcal{M}}$ and $p$  by
    \begin{align*}
        f_0(K) &= \left(\min_{1\le i<I^{\mathcal{M}}}\frac{A_{I^\mathcal{M}}-B_i}{K_{I^\mathcal{M}}-K_i}\right)(K-K_{I^\mathcal{M}})+A_{I^\mathcal{M}}, \\
    \tilde{\mathcal{M}} &= \mathcal{M}\cup\{f_0, f^D,0\},\\
    \hat{p} &= \max_{f\in\tilde{\mathcal{M}}} f.
    \end{align*}
    Then the function $\hat{p}$ is a non-decreasing convex function satisfying the following properties:
    \begin{enumerate}
  \item[a)] There exists $\epsilon > 0$ such that
    $\hat{p}(K) = 0$ for all $K \in [0,\epsilon]$.
    \item[b)] For any strike prices $K_1$ and $K_2$ with $K_1 < K_2$, the average rate of change of $\hat{p}$ is bounded above by $D$:
    \[ \frac{\hat{p}(K_2)-\hat{p}(K_1)}{K_2-K_1} \le D \]
    \item[c)] $B_n \leq \hat{p}(K_n) \leq A_n$ for all $n = 1,\dots,N$.
\end{enumerate}
\end{theorem}
\begin{remark}
    When $\mathcal{M}=\emptyset$, the pricing function can be constructed in the same manner as in Theorem \ref{thm1}.
\end{remark}
While we focused on constructing put option prices here, call option prices can be determined by a similar method. 
To this end, we first establish the definitions and notation for the call option market in a manner consistent with those used for put options.
\begin{definition}
     The system $(K_n, A_n, B_n)_{1\le{n}\le{N}}$ is called \textit{arbitrage-free} if the market consisting of zero-coupon bonds and call options corresponding to $(K_n, A_n, B_n)_{1\le{n}\le{N}}$ is free of arbitrage.
\end{definition}
Furthermore, we similarly define $f_{i,j}$ as the function representing the straight line passing through the two points $(K_i, A_i)$ and $(K_j, A_j)$. We then define the class of these functions, denoted by $\mathcal{M}'$, as follows:
\begin{align*}
&\mathcal{M}^\prime_0=\left\{ f_{i,j} \,;\, \exists m > i \text{ s.t. } f_{i,j}(K_m) < B_m \text{ and } \forall n, f_{i,j}(K_n) \le A_n\right\} \\
&\mathcal{L}^\prime_D=\left\{ f_{i,j}\,;\, f^\prime_{i,j}\ge -D\right\} \\
&\mathcal{M}^\prime=\mathcal{M}^\prime_0 \cap \mathcal{L}^\prime_D
\end{align*}
In addition, the functions $f^D_n, f^D$ and the indices $I^{\mathcal{M}^\prime},J^{\mathcal{M}^\prime}$ are defined as follows:
\begin{align*} 
    &f^D_i(K)=-D(K-K_i)+A_i, \quad \text{for }i\in\{1,\cdots,N\}, \\
    &f^D(K)=\min_{i\in\{1,\cdots,N\}}f^D_i(K), \\
    &I^{\mathcal{M}^\prime} = \min\{i \,;\, \exists j \text{ s.t. } i < j \le N, f_{i,j} \in \mathcal{M}^\prime \}, \quad
    J^{\mathcal{M}^\prime} = \max\{j \,;\, \exists i \text{ s.t. } 0 < i < j, f_{i,j} \in \mathcal{M}^\prime\}.
\end{align*}
Given these settings, the call option price function is formulated in the following manner:
\begin{theorem}\label{thm3}
Suppose that $\mathcal{M^\prime}\neq\emptyset$ and that the system $(K_n, A_n, B_n)_{1 \le n \le N}$ is arbitrage-free.
Define $f_0$, $\tilde{\mathcal{M}^\prime}$ and $c$ by
\begin{align*} 
f_0(K) &= \left(\max_{J^{\mathcal{M}^\prime}< j\le N}\frac{B_j-A_{J^{\mathcal{M}^\prime}}}{K_j-K_{J^{\mathcal{M}^\prime}}}\right)(K-K_{J^{\mathcal{M}^\prime}})+A_{J^{\mathcal{M}^\prime}},\\
\tilde{\mathcal{M}^\prime} &= \mathcal{M}^\prime\cup\{f_0,f^D,0\},\\
c & =\max_{f\in\tilde{\mathcal{M}^\prime}}f.
\end{align*}
Then, the function $c$ is a non-increasing convex function 
satisfying the following properties:
\begin{enumerate}
  \item[a)] There exists $\epsilon > 0$ such that
    $c(K) = 0$ for all $K \in [\epsilon,\infty)$.
    \item[b)] For any strike prices $K_1$ and $K_2$ with $K_1 < K_2$, the average rate of change of $c$ is bounded below by $-D$:
        \[ \frac{c(K_2)-c(K_1)}{K_2-K_1} \ge -D \]
    \item[c)] $B_n \leq c(K_n) \leq A_n$ for all $n = 1,\dots,N$.
\end{enumerate}
\end{theorem}
As established above, arbitrage-free prices can now be constructed for both types of options. This allows for the calculation of the volatility index through the direct evaluation of the integrals.

\subsection{Comparison with Existing Option Price Construction Methods}
In this section, we compare the proposed framework with existing methodologies for option price construction. Notably, \cite{SANOS} and \cite{LUCIC} have proposed construction methods based on optimization techniques. Specifically, \cite{SANOS} defines the option price function as a weighted combination of multiple Black-Scholes models and determines the function by optimizing these weights:
\[ C_1(k)=\sum w_j \text{BS}_{\text{call}}(k_j,k,\eta V_{\text{ATM}}) \]
where $\eta$ is a constant to adjust smoothness and $w_j$ are the weights to be optimized.
Similarly, \cite{LUCIC} determines option prices by optimizing the parameters of a spline interpolation, represented as:
\[ C_2(k)=\sum \tau_j N_j(k-k_j) \]
where $N_j$ denotes the basis spline functions and $\tau_j$ are the corresponding parameters.

The first fundamental difference lies in the reliance on a target price—typically the mid-price—for the optimization process. Existing methods are contingent upon these target values, whereas the proposed method eliminates such a dependency.

Secondly, \cite{SANOS} requires a significant amount of market data as inputs. For instance, it is necessary to identify the current forward price of the underlying asset $F_0$ to discount the price into a martingale with an expected value of unity. Furthermore, it requires the variance At-The-Money (ATM) as an argument for the Black-Scholes formula. The accumulation of these required inputs introduces the risk of estimation errors, which may ultimately lead to inaccuracies in the calculation of the volatility index. In contrast, the proposed method constructs the option price solely from the geometric configuration of bid and ask quotes, thereby circumventing these types of errors.

Thirdly, the spline interpolation employed in \cite{LUCIC} is performed only between observed strike prices and does not provide a mechanism for extrapolation. Consequently, this approach remains unable to address the underestimation of volatility caused by the truncation of the strike price range

In contrast to existing methodologies that prioritize the construction of smooth option price curves, the framework presented in this paper focuses on a price construction designed for model-free volatility estimation that remains strictly consistent with market observations.


\section{Computation of the volatility index}
\subsection{Comparison with the Cboe VIX}
We begin by conducting a comparative analysis with the VIX, the volatility index calculated by the Chicago Board Options Exchange (Cboe) for the S$\&$P 500 (SPX). For the VIX benchmarks, we employ values that we computed ourselves in accordance with the methodology outlined in the Cboe VIX White Paper \cite{Cboe1,Cboe2}.

The data used for these calculations are described as follows. We utilized the IvyDB database provided by OptionMetrics for option price data. For each trading date, expiration, option type (put or call), and strike price, this database provides "BestBid" and "BestOffer" values. These represent the highest closing bid and the lowest closing ask prices, respectively, across all exchanges where the underlying options are traded. We adopted these values as the bid and ask prices in our calculations.

Additionally, the database includes interest rate data. Although the Cboe methodology specifies the use of U.S. Treasury yield curve rates with appropriate interpolation, we adopted the interest rate data provided by IvyDB in this study to ensure consistency within the dataset.

\begin{figure}[t!]
  \centering
  \includegraphics[width=\linewidth]{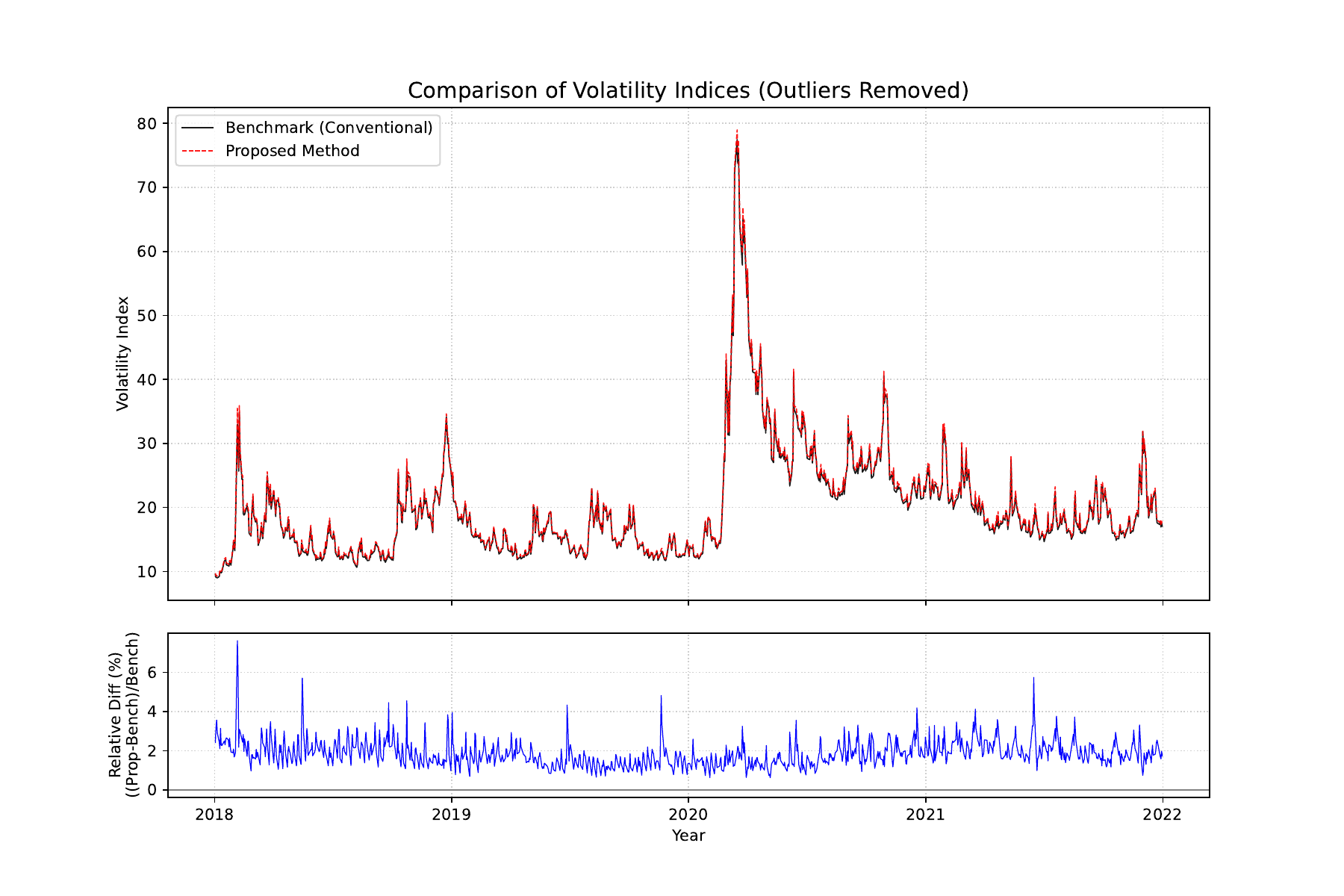}
  
  \caption{The upper panel displays the time series of the calculated indices, where the solid line represents the conventional method (Benchmark) and the dashed line represents the proposed method. The lower panel shows the relative divergence of the proposed method from the benchmark. Note that data points where either method exhibited anomalies (e.g., calculation failure or divergence) are excluded from this plot.}
  \label{fig:vix_comparison}
\end{figure}

We conducted calculations using two different methodologies over the period from 2018 to 2021. The results showed no significant divergence for the majority of the period, yielding nearly identical values. However, we observed instances of computational failure or divergence in both methods. Notably, there were significant differences in the characteristics of the data where these anomalies occurred.

First, the anomaly observed in the conventional method was the inability to compute values at all. This occurred on days with severely reduced market liquidity, such as during the COVID-19 shock. As will be detailed later, this was caused by the presence of multiple zero-bid data points near the ATM strike.

Next, the issue observed with the proposed method was a phenomenon where the calculated results diverged to extremely large values. However, in all cases where this occurred, the underlying market data implied the existence of arbitrage opportunities. Consequently, the assumption of Theorem \ref{thm2} was not satisfied; the constructed price function took strictly positive values at the origin, leading to the divergence of the integral. It is important to emphasize here that for option price data containing arbitrage opportunities, there exists no continuous arbitrage-free price function that is both consistent with the data and yields a convergent value. Therefore, given the nature of the market data, it is not unreasonable for the volatility index to exhibit abnormal values under these conditions. Furthermore, it is possible to construct a simple data filtering algorithm to exclude such arbitrage-violating data and compute normal values.

In contrast, the conventional method fails to compute if there are merely a few (as few as two) zero-bid data points, even though a zero-bid situation does not necessarily imply an arbitrage opportunity. Considering that this phenomenon tends to occur during periods of heightened financial risk, such as the COVID-19 shock, this represents a vulnerability regarding the robustness of the conventional method. Thus, it was confirmed that the proposed method demonstrates superior robustness compared to the conventional method.

\subsection{Case Study: Analysis of Calculation Failures}
In this section, we examine the specific conditions that led to a calculation failure in the conventional method, using the data from March 13, 2020, as a case study. To provide context, we first outline the data filtering process employed in the VIX calculation algorithm.

It is important to note that the Riemann sum approximation used for the VIX does not incorporate all available Out-of-the-Money (OTM) options. This exclusion is intended to filter out unreliable data associated with deep OTM options. The selection criteria are applied as follows: First, the At-The-Money (ATM) strike price is determined using put-call parity. Subsequently, candidates for inclusion are identified among strike prices lower than the ATM strike for put options (and conversely, higher strike prices for call options). The algorithm then selects strike prices moving outward from the ATM strike until two consecutive options with zero bids are encountered. Consequently, if two consecutive zero bids appear in the immediate vicinity of the ATM strike, the set of eligible strike prices becomes empty, rendering the VIX incalculable.

\begin{figure}[h!]
  \centering
  \includegraphics[width=\linewidth]{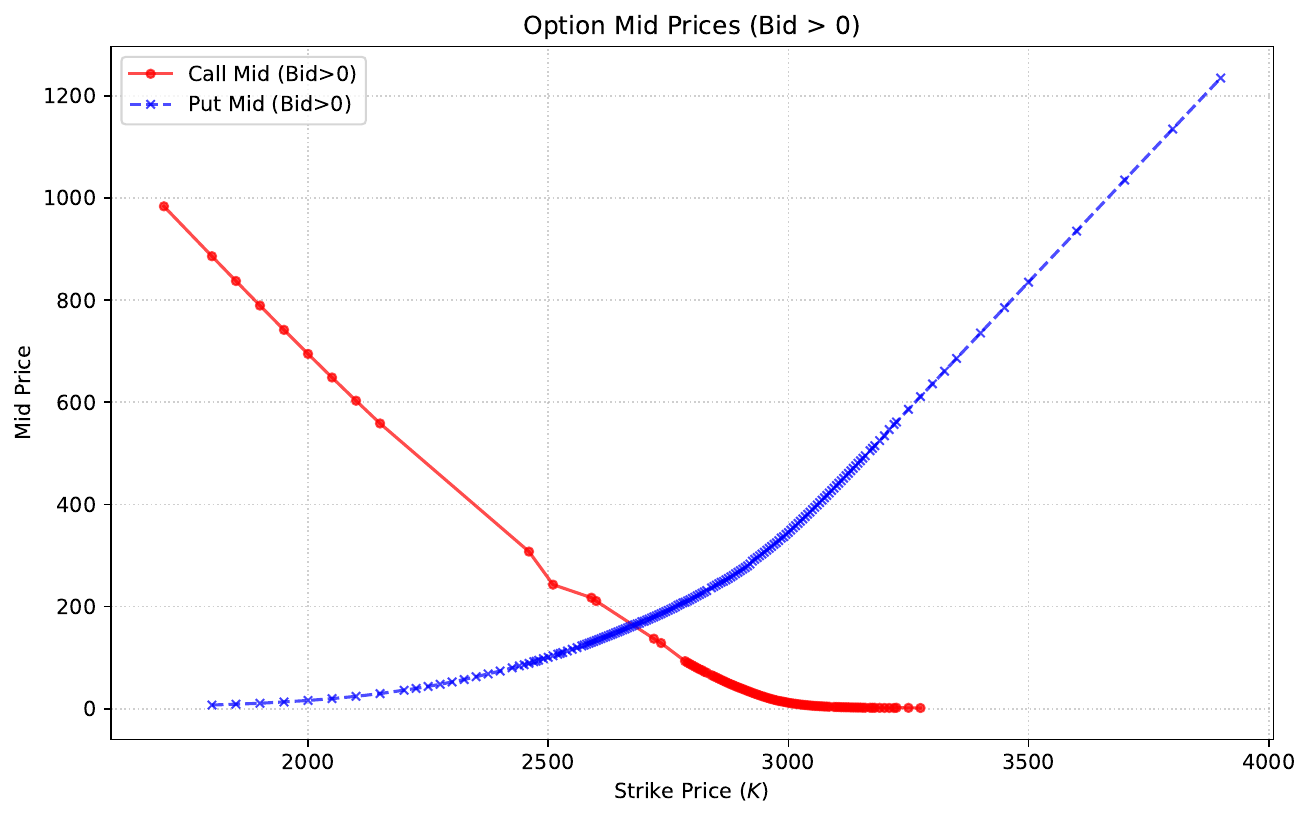}
  
  \caption{Option mid-prices for the Next-Term expiration on March 13, 2020. Options with zero bids are excluded, and valid data points are linearly interpolated.}
  \label{fig2}
\end{figure}

\begin{table}[h!]
  \centering
  \caption{Option Quotes near At-The-Money on March 13, 2020. Note the consecutive zero bids in Call options.}
  \label{tab:vix_failure_data}
  \begin{tabular}{rcccc}
    \toprule
    \textbf{Strike} & \textbf{Call Bid} & \textbf{Call Ask} \\
    \midrule
    2600 & 188 & 233 \\
    2605 & 0 & 230 \\
    $\vdots$ & $\vdots$ & $\vdots$ \\
    2675 & 0 & 185 \\
    2680 & 0 & 182 \\
    $\vdots$ & $\vdots$ & $\vdots$ \\
    2715 & 0 & 161 \\
    2720 & 115 & 158 \\
    2725 & 0 & 155 \\
    2730 & 0 & 152 \\
    \bottomrule
  \end{tabular}
  \label{tab1}
\end{table}

We now turn to an analysis of the data structure on the day of the failure.
Figure \ref{fig2} displays the mid-prices for options with strictly positive bids, along with their linear interpolation.
It is important to note that although the ATM strike lies near the intersection of the two curves, the distribution of call options around this critical point exhibits extreme sparsity.
Table \ref{tab1} confirms this observation: the occurrence of consecutive zero bids in the immediate vicinity of the ATM strike triggers the standard exclusion algorithm. This results in a premature truncation of the integration range, thereby making the VIX calculation impossible.

How, then, does the proposed method perform under such conditions?
As indicated by Theorem \ref{thm2}, the bid price in the proposed framework influences only $f_0$. Notably, the model remains unaffected by extremely low bid values, such as zero bids, particularly in the vicinity of the ATM strike. Consequently, it is possible to reconstruct the option price surface without any issues, even when the market data contains numerous zero bids (Figure 3). In this manner, the proposed method enables a robust evaluation of volatility even during financial crises—periods when market liquidity evaporates and financial risks are at their peak.

\begin{figure}[h!]
  \centering
  \includegraphics[width=\linewidth]{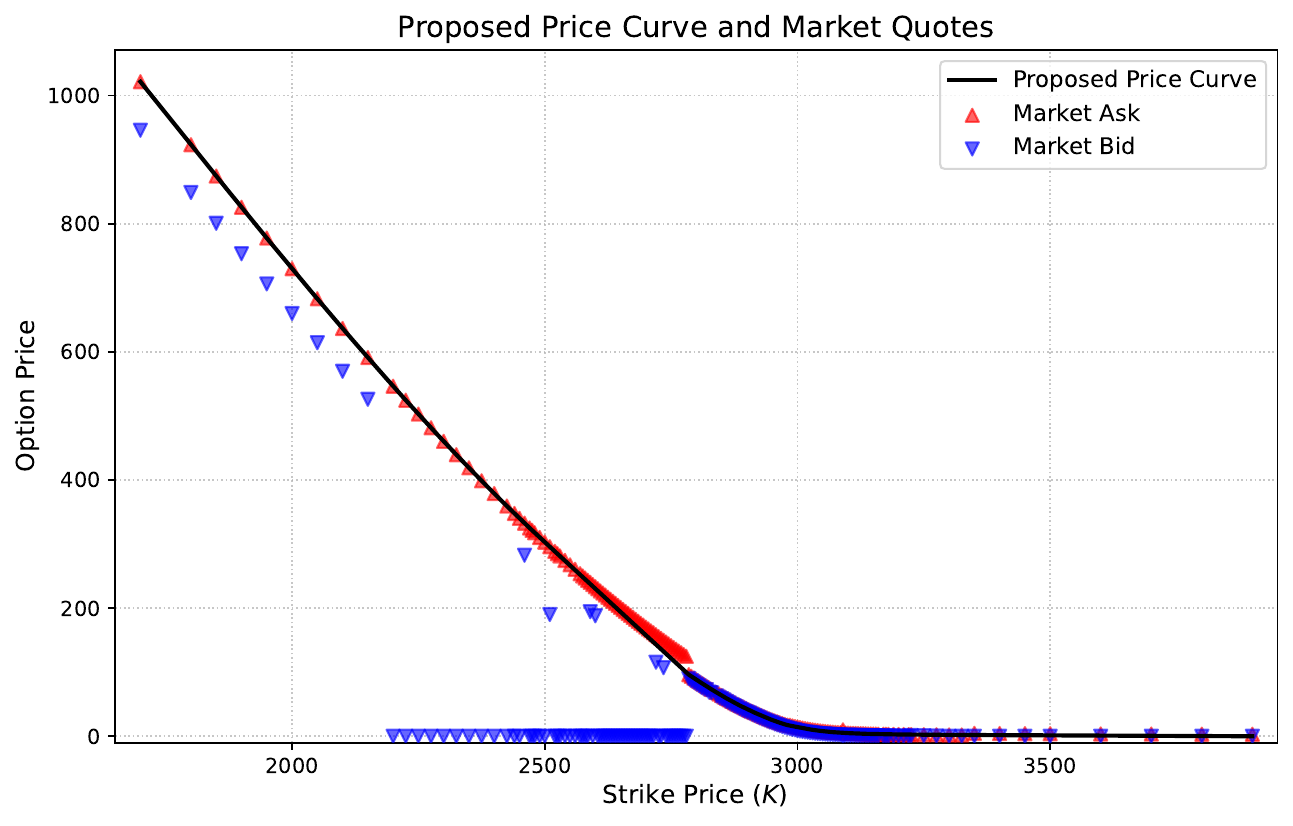}
  
  \caption{Market data and reconstructed option prices for call options on March 13, 2020. The proposed method successfully constructs arbitrage-free prices that are robust to the presence of zero bids.}
  \label{fig3}
\end{figure}

\subsection{Handling Arbitrage-Inconsistent Data}
In this section, we discuss the methodology for excluding data points in scenarios where market data contains arbitrage opportunities, which can cause the integration results to diverge.
In the proposed method, the divergence of the volatility index occurs when the integrand takes a positive value in the vicinity of the origin. The behavior near the origin is determined by the put option price function, where $f_0$ in Theorem \ref{thm2} plays a dominant role in governing its value. Consequently, the divergence of the integral is triggered by the presence of high bid prices or low ask prices in deep OTM put options. More specifically, this occurs when there exists a data set such that $g_{i, I^\mathcal{M}}(0) \ge 0$, which, according to Lemma \ref{lem3}, implies the existence of arbitrage opportunities. By excluding all data points associated with the indices $i$ where $g_{i, I^\mathcal{M}}(0) \ge 0$ as anomalous data related to arbitrage, the value of $f_0$ at the origin in Theorem 3.9 becomes strictly negative. This ensures that the integral for the volatility index converges to a finite value.
As an illustrative example, we present the near-term put options on February 18, 2021. As shown in Table 2, the strike price corresponding to the index $I^{\mathcal{M}}$ is 1500 with an ask price of 0.05. Notably, the bid price at the lower strike of 1475 is also 0.05, coinciding with $A_{I^{\mathcal{M}}}$. Consequently, $f_0$ becomes a constant function with a value of 0.05, leading to the divergence of the integral. Therefore, the data point at the strike price of 1475 is identified and excluded as an anomaly in this case.
\begin{table}[h!]
  \centering
  \caption{Selected option data for the near-term expiration on February 18, 2021.}
  \label{tab:vix_failure_data}
  \begin{tabular}{rcccc}
    \toprule
    \textbf{Strike} & \textbf{Put Bid} & \textbf{Put Ask} & Note\\
    \midrule
    $\vdots$ & $\vdots$ & $\vdots$ &\\
    1400 & 0 & 0.1 &\\
    1450 & 0 & 0.1 &\\
    1475 & 0.05 & 0.1 &\\
    1500 & 0.05 & 0.05 & $I^\mathcal{M}$\\
    \bottomrule
  \end{tabular}
  \label{tab1}
\end{table}

\end{document}